\documentclass{cccg16}
\usepackage{graphicx,amssymb,amsmath}

\usepackage{lscape}
\usepackage{subfigure,times,epsfig,epsf,latexsym,graphicx,multirow}
\usepackage{stfloats,amsmath,enumerate,setspace}
\usepackage{complexity}
\usepackage[ruled,vlined]{algorithm2e}
\usepackage{verbatim}
\usepackage{fancybox}
\usepackage{tabularx}
\usepackage{algorithmicx}
\usepackage{algpseudocode}
\usepackage{ifpdf}
\usepackage{color}
\usepackage{cite,url}
\DeclareGraphicsExtensions{.pdf,.jpg}
\usepackage{epstopdf}

\newtheorem{problem}{Problem}

\title{Problems on One Way Road Networks}

\author{Jammigumpula Ajaykumar\thanks{\tt Tezpur University, ak.jammi@gmail.com}
        \and
        Avinandan Das\thanks{\tt Tezpur University, adas33745@gmail.com}
        \and
        Navaneeta Saikia\thanks{\tt Tezpur University, saikia.navaneeta@gmail.com}
        \and
        Arindam Karmakar\thanks{\tt Tezpur University, karmarind@gmail.com}}

\index{Ajay Kumar, Jammigumpula}
\index{Das, Avinandan}
\index{Saikia, Navaneeta}
\index{Karmakar, Arindam}

\begin{document}
\thispagestyle{empty}
\maketitle

\begin{abstract}
Let $OWRN = \left\langle W_x,W_y \right\rangle$ be a One Way Road Network
where $W_x$ and $W_y$ are the sets of directed horizontal and vertical roads respectively.
$OWRN$ can be considered as a variation of directed grid graph. The intersections of the  horizontal and vertical roads are the vertices of $OWRN$ and any two consecutive vertices on a road are connected by an edge. In this work,  we analyze the problem of  collision free traffic configuration in a $OWRN$. A traffic configuration is a two-tuple  $TC=\left\langle OWRN,  C\right\rangle$, where  $C$ is a set of cars travelling on a pre-defined  path. We prove that finding a maximum cardinality subset $C_{sub}\subseteq C$ such that $TC=\left\langle OWRN,  C_{sub}\right\rangle$ is collision-free, is NP-hard. Lastly we investigate the properties of connectedness, shortest paths in a $OWRN$.
\end{abstract}

\section{Introduction}
The rapid development in the existing  motor vehicle technology
has led to the increase in demand of automated vehicles,
which are in themselves capable of various decision activities such as motion-controlling , path planning etc .This has motivated many to address a large number of algorithmic and optimisation problems.
The 1939 paper by Robbins\cite{rob}, which gives the idea of \textit{orientable} graphs and the paper by Masayoshi et.al \cite{masa}, are to a certain extent an inspiration to formulating our graph network. The work by Dasler and Mount \cite{arX}, which basically considers motion coordination of  a set of vehicles at a \textit{traffic-crossing} (intersection), has been a huge motivation and a much closer approach to that of ours. But unlike their work, we consider a much simpler version of a grid graph and mainly concentrate on analysing essential properties , and deriving suitable algorithms and structures to have a collision-free movement of traffic in the given graph network. Further, our work also mentions the possible shortest path configurations in our defined graph. We now discuss some definitions and notations that are referred to in the rest of the paper.

\subsection{Definitions and Notations}
A road is a directed line, which is either parallel to X-axis ($X_i$) or Y-axis ($Y_i$) and it is uniquely defined by its direction and distance from the corresponding parallel axis. Here direction is the constraint which restricts the movement on the road. 

Formally, a \textit{road} is defined as a 2-tuple, $X_i =\left\langle d_i,x_i \right\rangle$, $Y_j =\left\langle d_j,y_j \right\rangle$ , where $i,j$ are the indices with respect to their parallel axis, $d_k$ is the direction of the road i.e., $ d_k \in  \left\{{0,1}  \right\} $ (where 0 represents  $-ve$ direction and 1 represents $+ve$ direction of the respective axis) , and $x_i$ is the distance of the road $X_i$ from X-axis, similarly for $y_j$.

We define a \textit{One Way Road Network (OWRN)} as a network with a set of $n$ horizontal and $m$ vertical Roads. Formally a OWRN is a 2-tuple, $OWRN = \left\langle W_x,W_y \right\rangle$, 
where, $W_x = \left\lbrace X_1, X_2, X_3\dots, X_n \right\rbrace $, 
$W_y = \left\lbrace Y_1, Y_2, Y_3\dots, Y_m \right\rbrace $.

\begin{figure}[htb!]
\centering
\includegraphics[width=0.8\linewidth]{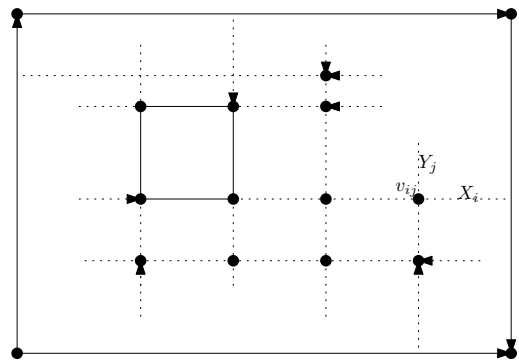}
\caption{Graphical representation of a $OWRN$}
\end{figure}

A \textit{junction} or \textit{vertex} $v_{ij}$ is defined as the intersection of $X_i$ and $Y_j$. Formally, $v_{ij} \in (W_x \times W_y)$.

An \textit{edge} is a connection between two adjacent vertices of a road and all the edges on a road are in the same direction as that of the road.


The \textit{boundary roads} of a OWRN are the outermost roads, i.e., $X_1, X_n, Y_1$ and $Y_m $. In this paper we term each vertex on the boundary roads as \textit{boundary vertex}, i.e., all the vertices with degrees 2 and 3.

A \textit{vehicle} $c$ is defined as a 3-tuple $\left\langle t,s,P \right\rangle$ , where $t$ is the starting time of the vehicle, $s$ is the speed of the vehicle , which is constant throughout the journey and each vehicle moves non-stop from its starting vertex to destination vertex (unless a collision occurs), and $P$ is the path to be travelled by the vehicle.

A \textit{path} $P_r$ of a vehicle $c_r$ is defined as the ordered set of vertices through which it traverses the OWRN, Formally, $P = \left\lbrace v_1, v_2, v_3\dots, v_l \right\rbrace$ , $\forall i$ , $v_{i} \in (W_x \times W_y)$,  and $\forall i, 0 < i < l $, $v_{i} \rightarrow v_{i+1}$.

A \textit{traffic configuration} is defined as a collection of vehicles over a OWRN. Formally a $TC$ is a 2-tuple, $TC=\left\langle OWRN,C \right\rangle$, where $C$ is the set of vehicles
$\left\lbrace c_1, c_2, c_3\dots,c_k \right\rbrace$.

Now, we define a \textit{collision} as two vehicles $c_i$ and $c_j$  $(i \neq j)$  reaching the same vertex orthogonally at the same time. So a \textit{collision-free} traffic configuration is a TC without any collisions.

\section{Results}
Before considering the traffic configuration problem, we define the connectivity of a OWRN.

\subsection{Connectivity of a One Way Road Network}
In this section, we consider a general OWRN of $n \times m$ roads, and show the conditions for it to be strongly-connected.
 
The reachability to (and from) the non-boundary vertices is evident from the following lemmas. 
\begin{lemma}
\label{lem:1}
For every non-boundary vertex $v_{ij}$, $1 < i < n$, $1 < j < m$ there exists $e$,$f$ such that  we can always reach the boundary vertices $v_{ie}$, $v_{fj}$ from $v_{ij}$.
\end{lemma}

\begin{proof}
We prove this lemma by considering the two roads which intersect to form the vertex $v_{ij}$.
\begin{enumerate}
\item
For the road $X_i$, if $d_i = 0$ then by definition we can reach $v_{i1}$ from $v_{ij}$, i.e., $e = 1$.
Otherwise we can reach $v_{in}$ from $v_{ij}$, i.e., $e = n$.
\item
For the road $Y_j$, if $d_j = 0$ then by definition we can reach $v_{1j}$ from $v_{ij}$, i.e., $f = 1$.
Otherwise we can reach $v_{mj}$ from $v_{ij}$, i.e., $f = m$.
\end{enumerate}
From the above conditions we can clearly see that for any non-boundary vertex $v_{ij}$, $\exists {e,f}$ such that $v_{ie}$, $v_{ej}$ are reachable from $v_{ij}$.

\end{proof}

\begin{lemma}
\label{lem:2}
For every non-boundary vertex $v_{ij}$ there exists $e$,$f$ such that $v_{ij}$ is reachable from the boundary vertices $v_{ie}$, $v_{fj}$.
\end{lemma}

\begin{proof}
The proof of this lemma is analogous to that of Lemma~\ref{lem:1}.\end{proof}

\begin{theorem}
A One Way Road Network is strongly-connected iff the boundary roads form a cycle.
\end{theorem}
\begin{proof}
The proof of this theorem follows from Lemma~\ref{lem:3} and Lemma~\ref{lem:4}.
\end{proof}

\begin{lemma}
\label{lem:3}
If all the boundary vertices of a OWRN form a cycle, then it is strongly-connected.
\end{lemma}

\begin{proof}
Given two vertices $v_{ij}$, $v_{kl}$ in a OWRN, to reach from $v_{ij}$ to $v_{kl}$, we have four different possibilities
\begin{enumerate}
\item
\textit{ Both boundary vertices}: Any boundary vertex is reachable from any other boundary vertex, since they all form a cycle. Therefore a path exists.
\item
\textit{$v_{ij}$ non-boundary vertex, $v_{kl}$ boundary vertex}: From \textit{Lemma~\ref{lem:1}} we know that, from any non-boundary vertex $v_{ij}$ we can always reach a boundary vertex, and from that vertex we can reach $v_{kl}$ as shown in $1$. Therefore a path exists.
\item
\textit{$v_{ij}$ boundary vertex, $v_{kl}$ non-boundary vertex}: From \textit{Lemma~\ref{lem:2}} we know that any non-boundary vertex $v_{kl}$ is always reachable from a boundary vertex, and which in turn is reachable from $v_{ij}$ as shown in $1$. Therefore a path exists.
\item
\textit{Both non-boundary vertices}: From $1$, $2$ and $3$ it is implied that there exists a path in this case too.

\end{enumerate}

\end{proof}

\begin{lemma}
\label{lem:4}
If a given One Way Road Network is strongly-connected, then all the boundary vertices form a cycle.
\end{lemma}

\begin{proof}
Let us assume on the contrary that the boundary vertices do not form a cycle in the OWRN.
Then there will exist a boundary vertex of degree 2 (corner vertex) such that either both the boundary roads are incoming or outgoing.
\begin{enumerate}
\item
\textit{Both incoming roads}: In this case, we will not be able to reach any other vertex from that vertex.
\item
\textit{Both outgoing roads}: In this case, we will not be able to reach that vertex from any other vertex. 
\end{enumerate}
Therefore, the OWRN is not strongly-connected.

Hence, by contradiction, we can claim that the boundary vertices of a strongly-connected OWRN will always form a cycle.
\end{proof}

\subsection{Traffic Configuration}
We now define the traffic configuration problem in a connected OWRN.
\begin{problem}
Given a traffic configuration $\left\langle OWRN,C \right\rangle $ , our objective is to find a maximum cardinality subset $C_{sub}$ , $C_{sub} \subseteq C$ , such that the new traffic configuration $\left\langle OWRN,C_{sub} \right\rangle $ is \textbf{collision-free}.
\end{problem}

In the following sections we discuss the hardness of the above problem, and also mention some of the restricted versions of the same.

\subsubsection{Hardness of Collision-Free Traffic Configuration}
In this section we show that finding a solution to the traffic configuration problem is \textit{\textbf{NP-Hard}}. For this , we have the following theorem.
\begin{theorem}
Given an undirected graph $G =\left\langle V,E \right\rangle $ , there exists a traffic configuration $\left\langle OWRN,C \right\rangle $ , computable in polynomial-time, such that the cardinality of Maximum Independent Set of $G$ is $k$ iff the maximum cardinality of $C_{sub}$ is $k$.
\end{theorem}

To prove this theorem, we reduce Maximum Independent Set problem to the Traffic Configuration problem, which is achieved with the help of the following lemmas and algorithms. 
 
\begin{lemma}
\label{lem:5}
For any complete graph $K_n$, there exists a traffic configuration $TC$, such that every vertex in $K_n$ has a respective car, and for every edge in $K_n$ there is a collision between the respective vehicles.
\end{lemma}
\begin{proof}
We prove this lemma using proof by construction. The following steps show how to construct a $TC$ from $K_n$.
\begin{enumerate}
\item
We construct a OWRN of $2n \times n$ roads, with $2n$ horizontal roads and $n$ vertical roads in which
\begin{enumerate}
\item
For the road $X_i$, $d_i = \begin{cases}
1 & 1 < i \leq 2n \\
0 & i = 1
\end{cases}$
\\and $x_i = \begin{cases}
0 & i = 1 \\
x_{i-1} + \delta & 1 < i \leq 2n
\end{cases}$
\item
For the road $Y_j$, $d_j = \begin{cases}
0 & 1 < j \leq n \\
1 & j = 1
\end{cases}$
\\and $y_j = \begin{cases}
0 & j = 1 \\
y_{j-1} + \delta & 1 < j \leq n
\end{cases}$
\end{enumerate}
where $\delta$ is a numeric constant.
\item
The set of vehicles $C$ is defined as $ \left\lbrace c_1, c_2, c_3\dots,c_n \right\rbrace$ and for each vehicle $c_i \in C$ we assume
\begin{enumerate}
\item
The start time to be $0$ and the velocity to be $\omega$.
\item
$P_i = \left\lbrace{v_{ri},v_{(r-1)i}\dots,v_{qi},v_{q(i+1)}\dots,v_{qn}}\right\rbrace$,  where $r = n+i-1$, $q = n-i+1$.
\item
$c_i = \left\langle 0,\omega,P_i \right\rangle$
\end{enumerate}
\item
Now we can observe that two vehicles $\left\lbrace c_i,c_j\right\rbrace$ $\in C$ collide at vertex $v_{(n-i+1)(j)}$ , where $i<j$.
\item
We assume that each node $l$ in $K_n$ corresponds to a vehicle $c_l$ , and each edge between two nodes $\alpha$ and $\gamma$ in $K_n$ corresponds to the collision of the respective vehicles $c_{\alpha}$,$c_{\gamma}$.  
\end{enumerate}
$\therefore$ We obtain the corresponding $TC = \left\langle OWRN, C \right\rangle $ of $K_n$.
\end{proof}

\begin{figure}[htb!]
\centering
\includegraphics[width=0.99\linewidth]{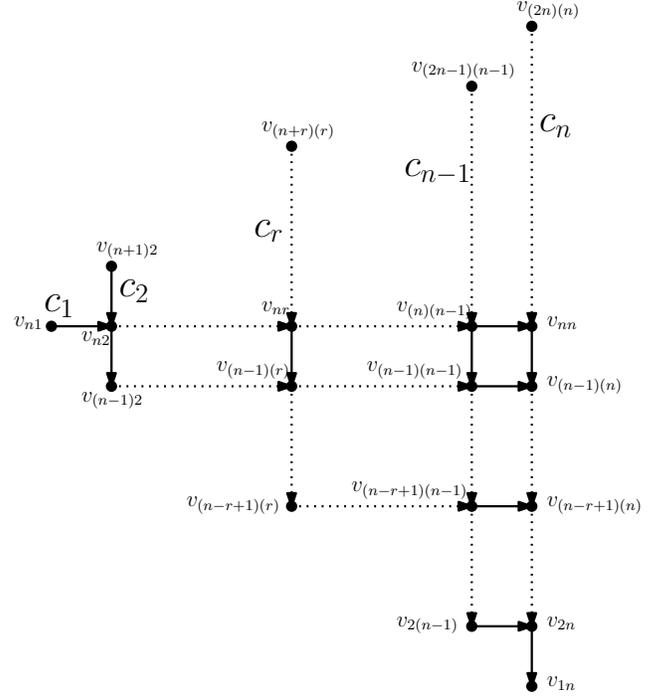}
\caption{Paths for vehicles $\left\lbrace c_1, c_2, \dots, c_n \right\rbrace$ in the $TC$ obtained from Lemma~\ref{lem:7}}
\end{figure}

Now to reduce any simple graph G, we first compute the corresponding TC for the complete graph $K_n(G)$. We then introduce 4 equi-spaced roads with directions (${d_k}$ 's) $\left\lbrace 0,1,0,1 \right\rbrace$ between every two adjacent roads $X_i$ ,$X_{i+1}$ and $Y_j$, $Y_{j+1}$, respectively, in the above formed OWRN, the path of each vehicle is to be modified accordingly.

We define method \Call{Delay}{$\alpha, \beta, P_i, \Delta$}, where $\alpha$ and $\beta$ are the two vertices in the path of $c_i$, and $\Delta$ is the total number of delays, which modify the path $P_i$ to introduce a $\Delta$ number of small time delays in between the vertices $\alpha$,$\beta$, this delay will also be propagated to all the successive vertices of $\beta$ in $P_i$. 

\begin{algorithm}
\caption{Delay method}\label{alg:delay}
Input: $P_r = (\dots,\alpha\dots,\beta\dots$), no.of delays $\Delta$\\
Output: $P_r$ after introduction of $\Delta$ delays
\begin{algorithmic}
\SetAlgoLined
\Procedure{Delay}{$\alpha,\beta,P_r,\Delta$}
	\State \If{$\Delta = 0$}{ \textbf{return}}
	\State $\gamma_1,\gamma_2$ are two successive vertices of $\alpha$ in $P_r$ 
	\State $\epsilon_1 \neq \gamma_2$, is a vertex $\mid$ there is an edge $\gamma_1 \to \epsilon_1$
	\State $\epsilon_2$ is a vertex $\mid$ $\epsilon_1 \to \epsilon_2$, $\epsilon_2 \to \gamma_1$
	\State $P_r =(\dots,\alpha,\gamma_1,\epsilon_1,\epsilon_2,\gamma_2\dots,\beta\dots)$
	\State \Call {Delay}{$\gamma_2,\beta,P_r,\Delta - 1$}
\EndProcedure
\end{algorithmic}
\end{algorithm}
\begin{figure}[htb!]
\centering
\includegraphics[width=0.78\linewidth]{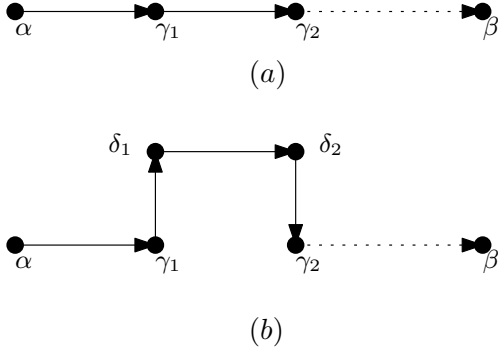}
\caption{ Path from $\alpha$ to $\beta$ { } $(a)$ Before delay introduction, $(b)$ After delay introduction}
\end{figure}

\begin{algorithm}
\caption{Collision method}\label{alg:collision}
\begin{algorithmic}
\SetAlgoLined
\Procedure{collisionVertex}{$c_i,c_j$}
	\State \If{$i = j $}{
		\State \textbf{return} \Call{collisionVertex}{$c_i,c_{j+1}$}
	}
	\State \ElseIf{$i = 0$} {\State \textbf{return} Starting Vertex of $c_j$}
	\State \textbf{return} Common vertex in $P_i$,$P_j$
\EndProcedure
\end{algorithmic}
\end{algorithm}

The method \Call{collisionVertex}{$c_i,c_j$}($i \neq j)$ will return the common vertex through which both the vehicles travel. In the base case($i = 0$) the Collision method returns the starting vertex of the vehicle $c_j$.

The following algorithm makes use of the above mentioned methods to construct the required TC by introducing some number of delays in the path of each vehicle.

\begin{algorithm}
\caption{Reduction algorithm}\label{alg:reduce}
\begin{algorithmic}
\SetAlgoLined
\Procedure{Reduce}{$c_i,c_j$}
	\State \If {$i > j$ or $i = 0 $} {
	\State \textbf {return} $0$
	}
	\State $\Delta$ $=$ \Call{Reduce}{$c_{i-1},c_j$}
	\State \Call {Reduce}{$c_i$,$c_{j-1}$}
	\State $\alpha$ = \Call{collisionVertex}{$c_{i-1},c_j$}
	\State $\beta$ = \Call {collisionVertex}{$c_i,c_j$}
	\State \If{\Call{hasEdge}{$i,j$}}{
	\State \Call{Delay}{$\alpha,\beta,P_j,i-\Delta$}	
	\State \textbf{return} $i$
	}
	\State \ElseIf{$i - \Delta > 1$}{
	\State \Call{Delay}{$\alpha,\beta,P_j,i-(\Delta+1)$}
	\State \textbf{return} $i-1$
	}
	\State \textbf{return} $\Delta$
\EndProcedure
\end{algorithmic}
\end{algorithm}

The reduction algorithm is constructed using the following properties:
\begin{enumerate}[{Property} 1:]
\item
The number of delays introduced in the path of a vehicle $c_i$ is equal to $i$.
\item
If there is an edge between two nodes $i,j$, $j > i$ in G, then $c_i$ and $c_j$ will have collision in the TC. The number of delays introduced in the path $P_j$ before the collision of $c_i$,$c_j$ is $i$.
\item
If there is no edge between two nodes $i,j$, $j > i$ in G, then $c_i$ and $c_j$ will not have a collision in the TC. The number of delays introduced in the path $P_j$ before the collision of $c_i$,$c_j$ is $i-1$.
\end{enumerate}

The method \Call{hasEdge}{$i,j$} will return value true if there is an edge between $i$ and $j$ in the graph G, else false.
\begin{lemma}
\label{lem:6}
The maximum number of delays introduced between the two collision vertices $\alpha$ and $\beta$ as defined in the reduction algorithm, will be two.
\end{lemma}

\begin{proof}
The proof of this lemma follows from the above stated properties.
The number of delays introduced in the path $P_j$, before collision of vehicles $c_i$ and $c_j$ is either $i$,$i-1$.
The number of delays introduced in the path $P_j$, before collision of vehicles $c_{i+1}$ and $c_j$ is either $i+1$,$i$.

$\therefore$ the maximum number of delays that can be introduced between $\alpha$ and $\beta$ is two.
\end{proof}

From the above Lemma and the reduction algorithm, we have the following Lemma
\begin{lemma}
\label{lem:7}
The above Reduction algorithm can be solved using Dynamic Programming approach in polynomial-time $\mathcal{O}(n^2)$, and the space complexity of both TC and OWRN created is $\mathcal{O}(n^2)$. 
\end{lemma} 

\begin{lemma}
\label{lem:8}
If $C_{sub}$ be any subset of C in TC such that $TC_{new} = \left\langle OWRN,C_{sub} \right\rangle $ is collision-free, then $C_{sub}$ corresponds to Independent Set of G.
\end{lemma}

\begin{proof}
Since $TC_{new}$ is collision-free, so no two nodes in the graph G, which corresponds to respective cars in $C{sub}$,  consists of an edge. Thus, we can claim that $C_{sub}$ corresponds to an independent set in G.
\end{proof}

From \textit{Lemma~\ref{lem:8}} we can say that maximum $C_{sub}$ corresponds to \textit{Maximum Independent Set} in $G$. Now, using \textit{Lemma~\ref{lem:7}} and \textit{Lemma~\ref{lem:8}} we can prove that the traffic configuration problem is \textit{NP-Hard}.

\subsubsection{Restricted Version}
If we constrain our vehicles to move in a straight line motion, then the corresponding graph to $TC$ will be a Bipartite Graph. And \textit{Maximum Independent Set} of a Bipartite Graph can be computed using Konig's Theorem and Network-Flow Algorithm in polynomial-time. Hence, the restricted version of the problem is solvable in polynomial-time.

\subsection{Shortest Path Properties}
Suppose in a city of only One Way Road Network ,a person wants to travel from one point to another with minimum distance. Now, the objective would be to compute the shortest path to the destination in least possible time. Designing efficient algorithms to compute the shortest path in a One Way Road Network would be useful in many applications in the areas of facility location, digital micro-fluidic bio-chips,etc. \\
The length of the shortest path between two vertices in a OWRN may not be the Manhattan distance. There may be a pair of neighbouring vertices which are the farthest pair of vertices in the OWRN metric. A \textit{turn} in a path is defined when two consecutive pair of edges are from different roads. We have the following Lemmas for shortest path between any two vertices in a OWRN:

\begin{lemma}
\label{lem:9}
Between any pair of vertices $u$, $v$ , there exists a shortest path of at most four turns.
\end{lemma}

\begin{proof}
We can observe that a shortest path of $k$ turns is geometrically similar to an addition of one more turn at the end of shortest path of $k-1$ turns, And the path with $k$ turns can be reduced if it's sub-path with $k-1$ turns can be reduced.

\begin{enumerate}
\item The path with zero turns is a straight line and the path with one turn is an L shaped, which are trivial.
\item The path with two turns, two configurations $(c),(d)$ are possible and are valid as shown in \textit{Figure~\ref{fig:6}}.
\item The path with three turns, two configurations$(e),(f)$ shown in \textit{Figure~\ref{fig:6}} are valid and the other two shown below are not valid as there exist another shortest path with 1 turn.
\begin{figure}[htb!]
\label{fig:4}
\centering
\includegraphics[width=0.85\linewidth]{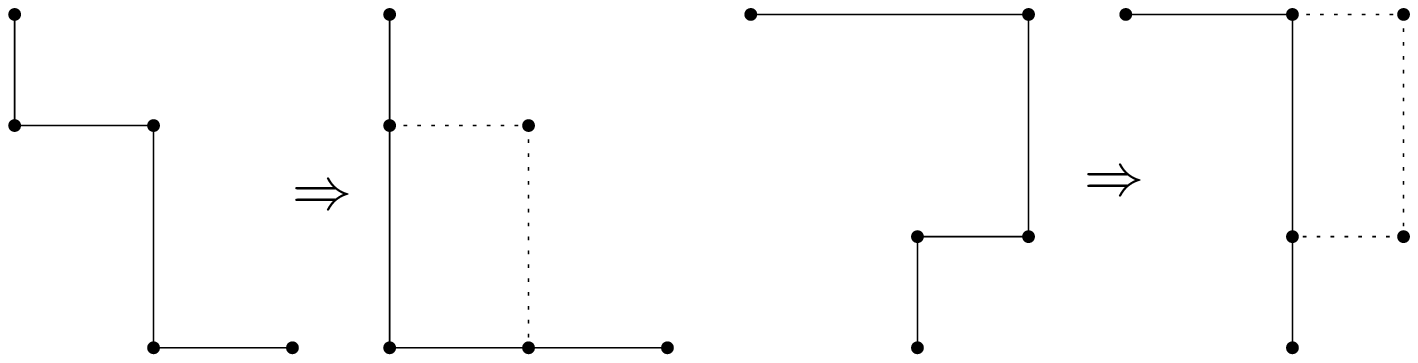}
\end{figure}
\item The path with four turns, from the above explanation it can be easily seen that only $(g),(h)$ are valid configurations with four turns.
\item The path with five turns, in the below figure we can see that a five turn path will become either a one turn or four turn.
\begin{figure}[htb!]
\label{fig:5}
\centering
\includegraphics[width=0.85\linewidth]{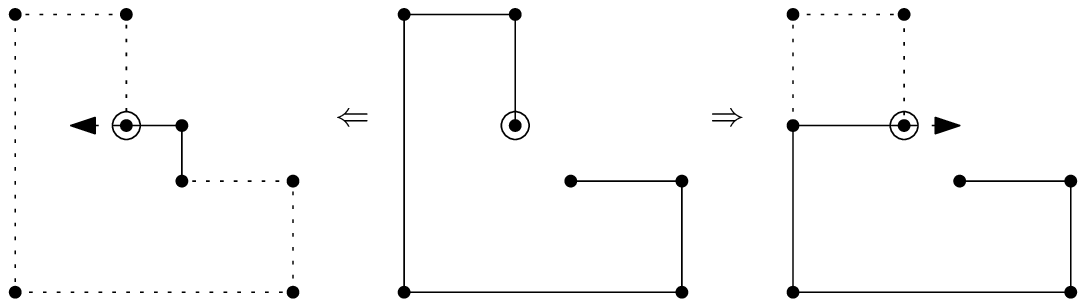}
\end{figure}
similarly we can see that for any five turn path there exist a shortest path with less number of turns.
\end{enumerate}
Hence we can see that there always exist a shortest path with at most 4 turns.
\end{proof}

From \textit{Lemma~\ref{lem:9}} we observe that there exist a shortest path between every pair of vertices in the OWRN which will be a rotationaly symmetric to one of the paths shown below:
\begin{figure}[htb!]
\label{fig:6}
\centering
\includegraphics[width=0.95\linewidth]{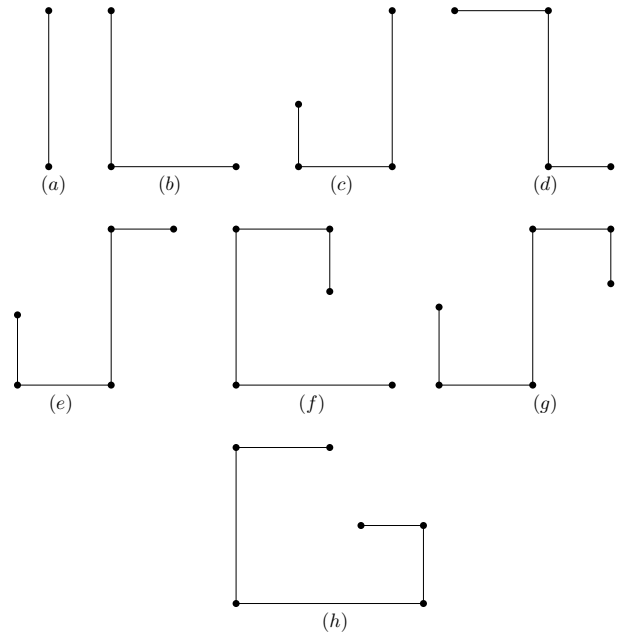}
\caption{Shortest path configurations}
\end{figure}

\begin{lemma}
The upper bound on the length of the shortest path between any pair of vertices $u$, $v$ is the perimeter of the boundary of the OWRN.
\end{lemma}

\begin{proof}
The proof of this lemma follows from \textit{Lemma~\ref{lem:1}}, \textit{Lemma~\ref{lem:2}} and \textit{Lemma~\ref{lem:9}}.

\end{proof}

\section{Remarks}
We have shown all the possible configurations of the path, that connects two vertices in a OWRN. In the future we will extend this work to compute various kind of facility location problems on a OWRN. It will be interesting to investigate the time complexity of one-centre or k-centre problems with respect to OWRN metric. Other interesting problems may be to design an efficient data structure for dynamic maintenance of shortest path in directed grid graphs.

\small
\bibliographystyle{abbrv}


\end{document}